%% file: am08-long.tex
\documentclass[fleqn,12pt]{article}
\usepackage{url}
\usepackage{amsmath}
\usepackage{amsfonts}
\usepackage{amssymb}
\usepackage{thanks}

\usepackage{color}
\usepackage{graphics}

\usepackage{latexsym}

\include{ecmds}
\include{header}
  \newenvironment{proof}
        {\medskip\noindent{\bf Proof.}}
        {\hfill$\Box$\medskip}

\title{Optimal Strategies in Sequential Bidding}

\author{Krzysztof R. Apt \thanksref{t1}\thanksref{t2}\\
\and \mbox{Vangelis Markakis \thanksref{t1}\thanksref{t3}
\thankstext{t1}{CWI, Kruislaan 413, 1098 SJ Amsterdam, The Netherlands}
\thankstext{t2}{Institute of Logic, Language and Computation, University of Amsterdam}
\thankstext{t3}{current affiliation: Athens University of Economics and Business, Athens, Greece}
}}

\date{}
\begin{document}
\maketitle

\begin{abstract}
  We are interested in mechanisms that maximize social welfare.  In
  \cite{ACGM08} this problem was studied for multi-unit auctions
  with unit demand bidders and for the public
  project problem, and in each case social welfare undominated
  mechanisms in the class of feasible and incentive compatible
  mechanisms were identified.
  
  One way to improve upon these optimality results is by allowing the
  players to move sequentially.  With this in mind, we study here
  sequential versions of two feasible Groves mechanisms used for
  single item auctions: the Vickrey auction and the Bailey-Cavallo mechanism.
  
  Because of the absence of dominant strategies in this sequential
  setting, we focus on a weaker concept of an optimal strategy. For
  each mechanism we introduce natural optimal strategies and observe
  that in each mechanism these strategies exhibit different behaviour.
  However, we then show that among all optimal strategies, the one we
  introduce for each mechanism maximizes the social welfare when each
  player follows it. The resulting social welfare can be larger than
  the one obtained in the simultaneous setting. Finally, we show that,
  when interpreting both mechanisms as simultaneous ones, the vectors
  of the proposed strategies form a Pareto optimal Nash equilibrium in
  the class of optimal strategies.
\end{abstract}

\section{Introduction}
\label{sec:intro}

\subsection{Motivation}

Suppose that a group of agents would like to determine who among them
values a given object the most. A natural way to approach this problem
is by viewing it as a single unit auction.  Such an auction is
traditionally used as a means of determining by a seller to which
bidder and for which price the object is to be sold. But the
underlying mechanism can also be used in the situation we are
interested in, since in both cases the requirement of incentive
compatibility, that is, of preventing manipulations, remains the same.
On the other hand, the absence of a seller changes the perspective and
leads to different considerations. Instead of maximizing the revenue
of the seller we are then interested in maximizing the social welfare,
that is in determining the winner at a minimal cost.

This brings us to the problem of finding incentive compatible
mechanisms that are optimal in the sense that no other incentive
compatible mechanism generates a larger social welfare.  Recently, in
\cite{ACGM08} this problem was studied for two domains: multi-unit
auctions with unit demand bidders and the public project problem of
\cite{Cla71}.  It was proved that for the first domain the
optimal-in-expectation linear (OEL) redistribution mechanisms,
introduced in~\cite{GC08b}, are optimal, while for the second one the
pivotal mechanism is optimal. 

\subsection{Sequentiality}

One way to improve upon these optimality results is by relaxing the
assumption of simultaneity and allowing the players to move
sequentially.  This set up was studied in \cite{AE08} for various
versions of the public project problem and it was shown that in
the sequential pivotal mechanism natural strategies exist that allow
the players to increase the social welfare that would be generated if
they moved simultaneously. In this paper we consider such a modified
setting for the case of single unit auctions.  We call it
\oldbfe{sequential bidding} as the concept of a `sequential auction'
usually refers to a sequence of auctions, see, e.g. \cite[ chapter
15]{Kri02}.

So we assume that there is a single object for sale and the
players announce their bids sequentially. In contrast to the open cry
auctions each player announces his bid \emph{exactly once}. Once all
bids have been announced, a mechanism is used to allocate the
object to the highest bidder and determine his payments, to a bank
and possibly to other players.  Such a sequential setting is very natural
when we wish to determine which agent values a given object the
most: the agents then simply state in a random order their valuations.

We study here sequential versions of two incentive compatible
mechanisms for single unit auctions, Vickrey auction and
Bailey-Cavallo mechanism of \cite{Bai97} and \cite{Cav06}, the latter
being a simplest and most natural mechanism in the class of OEL
mechanisms. In the former mechanism the payment consists of the second
highest bid and is sent to the bank, while in the latter the payments
are generated both to the bank and to other players.

\subsection{Results}

We first show that in the sequential Vickrey auction and sequential
Bailey-Cavallo mechanism no \emph{dominant strategies} exist (except
for the last player). Therefore we settle on a weaker concept, that of
an \emph{optimal strategy}. An optimal strategy for a player $i$
yields for all type vectors a best response to all joint strategies of
the other players, under the assumption that the players that follow
$i$ are myopic (i.e., their strategy does not depend on the types of
the previous players). For example truth telling is a myopic strategy.
These strategies exhibit different behaviour in these two mechanisms.
In sequential Vickrey auction optimal strategies yield the same payoff
against each vector of optimal strategies of the other players, which
is not the case in the sequential Bailey-Cavallo mechanism.

In both mechanisms we propose natural optimal strategies that differ
from truth-telling.  While truth-telling strategy focuses only on
player's own payoff, the proposed strategies are additionally `good'
for the society: in both mechanisms they yield the maximal social
welfare, when each player follows it. In addition, the same outcome is
realized under these optimal strategies as under truth-telling.

In the sequential Vickrey auction this strategy is also
\emph{socially maximal}, which means that it simultaneously guarantees
the maximal utility to every other player, under the assumption that
they are truth-telling.  In contrast, in the sequential Bailey-Cavallo
mechanism no socially maximal strategy exists, except for the first
and last player. (We actually establish a stronger negative result.)
The nature of the proposed strategies can be further clarified when
{\it interpreting} both mechanisms as simultaneous ones. We show that in
both cases their vectors form then a Pareto optimal Nash equilibrium
in the class of optimal strategies.

The paper is organized as follows. We first recall
Groves mechanisms, in particular Vickrey auction and Bailey-Cavallo
mechanism. We review in Section \ref{sec:sequential1} the
concepts introduced in \cite{AE08} concerned with a taxonomy of
strategies in sequential mechanisms.  In Section \ref{sec:seq-groves}
we establish some preparatory results for sequential Groves
mechanisms.

The main results are established in the next two sections,
\ref{sec:vic} and \ref{sec:bc}, in which we deal with the sequential
versions of Vickrey auction and Bailey-Cavallo mechanism. In Section
\ref{sec:nash} we discuss the corresponding
simultaneous mechanisms, and in Section \ref{sec:conc} 
discuss possible future research.


\section{Preliminaries}\label{sec:prelim}

We collect here the necessary background material.
Readers familiar with Groves mechanisms can safely move to Subsection \ref{subsec:two}.

\subsection{Tax-based mechanisms}

We first briefly review tax-based mechanisms (see, e.g.,~\cite{MWG95}).
Assume that there is a finite set of possible outcomes or \oldbfe{decisions} $D$,
a set $\{1, \LL, n\}$ of players where $n \geq 2$, and for each player $i$
a set of \oldbfe{types} $\Theta_i$ and an (\oldbfe{initial})
\oldbfe{utility function} $ v_i : D \times \Theta_i \myra \mathbb{R}$. 
A \oldbfe{decision rule} is a function $f: \Theta \myra D$, where
$\Theta := \Theta_1 \times \cdots \times \Theta_n$.  

In a \oldbfe{tax-based mechanism}, in short a \oldbfe{mechanism}, each
player reports a type $\theta_i$ and based on this, the mechanism
selects an outcome and a payment to be made by every agent. Hence a
mechanism is given by a pair of functions $(f,t)$, where $f$ is the
decision rule and $t = (t_1,...,t_n)$ is the tax function that
determine the decision taken and the players' payments given the
reported types $\theta_1, \LL, \theta_n$, i.e., $f: \Theta \myra D$,
and $t: \Theta \myra \mathbb{R}^n$.

We assume that the (\oldbfe{final}) \oldbfe{utility function} for
player $i$ is a function $u_i: D \times \mathbb{R}^n \times \Theta_i
\myra \mathbb{R}$ defined by $ u_i(d,t_1, \LL, t_n, \theta_i) :=
v_i(d, \theta_i) + t_i $.  For each vector $\theta$ of announced
types, if $t_i(\theta) \geq 0$, player $i$ \oldbfe{receives}
$t_i(\theta)$, and if $t_i(\theta) < 0$, he \oldbfe{pays}
$|t_i(\theta)|$.  Thus when the true type of player $i$ is $\theta_i$
and his announced type is $\theta'_i$, his final utility is
\[
u_i((f,t)(\theta'_i,
\theta_{-i}), \theta_i) = v_i(f(\theta'_i, \theta_{-i}), \theta_i) +
t_i(\theta'_i, \theta_{-i}),
\] 
where $\theta_{-i}$ are the types
announced by the other players.

We say that a mechanism $(f,t)$ is
  \begin{enumerate}
  \item[$\bullet$] \oldbfe{efficient} if for all $\theta \in
\Theta$, $f(\theta)\in {\rm argmax}_{d \in D} \sum_{i=1}^n v_i(d,\theta_i)$, 
i.e., the taken decision maximizes the initial social welfare,

  \item[$\bullet$] \oldbfe{feasible} if for all $\theta$, $\sum_{i = 1}^{n} t_i(\theta) \leq 0$, i.e., the mechanism does not need to be funded by an external source,

  \item[$\bullet$] \oldbfe{incentive compatible}
if for all $\theta$, $i \in \C{1,\LL,n}$ and
$\theta'_i$,
\[
u_i((f,t)(\theta_i, \theta_{-i}), \theta_i) \geq
u_i((f,t)(\theta'_i, \theta_{-i}), \theta_i).
\]
  \end{enumerate}

\subsection{Groves mechanisms}

Each \oldbfe{Groves mechanism} is an efficient tax-based mechanism
$(f,t)$ such that the following hold for all $\theta \in
\Theta$\footnote{Here and below $\sum_{j\not=i}$ is a shorthand for
  the summation over all $j \in \{1,\LL,n\}, \ j \not=i$ and similarly
  for $\max_{j\not=i}$.}:

\begin{itemize}

\item[$\bullet$] $t_i : \Theta \myra \mathbb{R}$ is defined by $t_i(\theta) := g_i(\theta) + h_i(\theta_{-i})$, where

\item[$\bullet$] $g_i(\theta) := \sum_{j \neq i} v_j(f(\theta), \theta_j)$,
  
\item[$\bullet$] $h_i: \Theta_{-i} \myra \mathbb{R}$ is an arbitrary function.

\end{itemize}

Intuitively, $g_i(\theta)$ represents the initial social welfare from the
decision $f(\theta)$, when player $i$'s (initial) utility is not counted.  
Recall now the following crucial result, see
e.g.,~\cite{MWG95}.  
\II

\NI \textbf{Groves Theorem} Every Groves mechanism $(f,t)$ is incentive compatible.  \II


A special Groves mechanism, called the \oldbfe{pivotal mechanism}\footnote{This is sometimes referred to as the VCG mechanism.}, is obtained using
$ h_i(\theta_{-i}) := - \max_{d \in D} \sum_{j \neq i}
v_j(d,\theta_j)$.
In this case, the tax, $t_i^{p}(\theta)$, is defined by
\[
t_i^{p}(\theta)  := \sum_{j \neq i} v_j(f(\theta), \theta_j) - \max_{d \in D} \sum_{j \neq i} v_j(d, \theta_j),
\]
which shows that the pivotal mechanism is feasible.

\subsubsection{Groves mechanisms for single item auctions}

Given a sequence $a := (a_1, \LL, a_j)$ of
reals we denote the least $l$ such that $a_l = \max_{k \in \{1,
  \ldots, j\}} a_k$ by $\textrm{argsmax}_{k \in \{1, \ldots, j\}}
a_k$ or simply by $\textrm{argsmax} \: a$.
A \oldbfe{single item sealed bid auction}, in short an \oldbfe{auction},
is modelled by choosing
\begin{itemize}

\item $D = \{1, \ldots, n\}$,

 \item each $\Theta_i$ to be the set ${\cal R}_+$ of non-negative reals;
$\theta_i \in \Theta_i$ is player $i$'s valuation of the object,

\item 
$
        v_i(d, \theta_i) := 
        \left\{
        \begin{array}{l@{\extracolsep{3mm}}l}
        \theta_i   & \mathrm{if}\  d = i \\
        0      & \mathrm{otherwise}
        \end{array}
        \right.
$

\item 
$
f(\theta) := \textrm{argsmax} \: \theta.
$

\end{itemize}

Here decision $d \in D$ indicates which player is the winner, that is
the player to whom the object is sold.  Hence the object is sold to
the highest bidder and in the case of a tie we allocate the object to
the player with the lowest index.\footnote{If we make a different
  assumption on breaking ties, some of our proofs need to be adjusted,
  but similar results hold.}  By the choice of $f$ each auction
mechanism $(f,t)$ is efficient.

Note that player's $i$ final utility 
in an auction mechanism $(f,t)$ is
\[
        u_i((f,t)(\theta'_i, \theta_{-i}), \theta_i) =
        \left\{
        \begin{array}{l@{\extracolsep{3mm}}l}
        \theta_i + t_i(\theta'_i, \theta_{-i})  & \mathrm{if}\  \textrm{argsmax}\: \theta' = i \\
        t_i(\theta'_i, \theta_{-i})     & \mathrm{otherwise}
        \end{array}
        \right.
\]
where player's $i$ received type is $\theta_i$, his
announced type is $\theta'_i$, $\theta_{-i}$ is the vector of types announced by other players and $\theta' = (\theta'_i,\theta_{-i})$.

By a \oldbfe{Groves auction} we mean a Groves mechanism for an auction setting.
Below, given a sequence $s$ of reals we denote by $\theta^*$ its reordering in descending order.
Then $\theta^*_k$ is the $k$th largest
element in $\theta$. For example, for $\theta = (1,5,0,3,2)$ we have 
${(\theta_{-2})}^*_2 = 2$ since $\theta_{-2} = (1,0,3,2)$.

The \oldbfe{Vickrey auction} is the pivotal mechanism for a single item auction. 
Hence it uses the following taxes:
\[
        t^{p}_i(\theta) := 
        \left\{
        \begin{array}{l@{\extracolsep{3mm}}l}
        - \theta^*_{2}   & \mathrm{if}\  \textrm{argsmax} \: \theta = i \\
        0      & \mathrm{otherwise}
        \end{array}
        \right.
\]
That is, the winner pays the second highest bid. 

\subsubsection{The Bailey-Cavallo mechanism}

The \oldbfe{Bailey-Cavallo} mechanism, in short \oldbfe{BC auction}, is the
mechanism originally proposed in \cite{Bai97} and \cite{Cav06} (in fact,
Bailey's mechanism is not always the same as Cavallo's mechanism, but it is
in the setting in which we study it).
To define it note that each Groves mechanism is uniquely determined by a
\oldbfe{redistribution function} $r := (r_1, \LL, r_n)$, where each
$r_i: \Theta_{-i} \myra \mathbb{R}$ is an arbitrary function. Given a
redistribution function $r$ the tax for player $i$ is defined by
$t_i(\theta) := t^{p}_i(\theta) + r_i(\theta_{-i})$, i.e., we can think of a Groves mechanism as first running the pivotal mechanism and then redistributing the pivotal taxes.

The BC auction is a Groves
mechanism defined by using the following redistribution function $r :=
(r_1, \LL, r_n)$ (to ensure that it is well-defined we need to assume
that $n \geq 3$):
\[
r_i(\theta_{-i}) := \frac{{(\theta_{-i})}^*_2}{n}
\]
that is, by using
$
t_i(\theta) := t^{p}_i(\theta) + r_i(\theta_{-i})
$.

The BC auction is feasible since 
for all $i \in \{1, \LL, n\}$ and $\theta$ we have $(\theta_{-i})^*_2 \leq \theta^*_2$ and as a result
\[
\sum_{i = 1}^{n} t_i(\theta)  = \sum_{i = 1}^{n} t^{p}_i(\theta) + \sum_{i = 1}^{n} r_i(\theta_{-i}) = 
\sum_{i = 1}^{n} \frac{- \theta^*_2 + (\theta_{-i})^*_2}{n} \leq 0.
\]

Furthermore, given the sequence $\theta$ of submitted types, note that
if player $i$ is the first or the second highest bidder, then
${(\theta_{-i})}^*_{2} = {\theta}^*_{3}$.
For the rest of the players, ${(\theta_{-i})}^*_{2} = {\theta}^*_{2}$.



Hence
\begin{equation}\label{eq:bc-redistr}
\sum_{i = 1}^{n} r_i(\theta_{-i}) = \sum_{i = 1}^{n} \frac{{(\theta_{-i})}^*_2}{n} =
\frac{n-2}{n} {\theta}^*_{2} + \frac{2}{n} \theta^*_{3},
\end{equation}
and
\begin{equation}\label{eq:bc-tax}
\sum_{i = 1}^{n} t_i(\theta) = \frac{2}{n} (\theta^*_{2} - \theta^*_{3}) 
\end{equation}
so when the second highest type, $\theta^*_{2}$, is strictly positive
and the third highest type, $\theta^*_{2}$, is non-negative, the BC
auction yields a strictly higher social welfare than the pivotal
mechanism. Note also that the aggregate tax is $0$ when the second
highest and the third highest bids are the same.


In some situations it is useful to employ Groves auctions that
maximize the final social welfare.  This is for example the case when,
as discussed in the introduction, the
players want to determine at a minimal cost, using an incentive compatible
mechanism, who among them values a given object most.  For such
applications by the above observation Vickrey auction is not an
appropriate mechanism. In \cite{ACGM08} it was proved that the BC auction
is an appropriate mechanism in the sense that no Groves auction

\begin{itemize}
\item always generates a larger or equal social welfare than BC,

\item sometimes generates a strictly larger social welfare than BC.
\end{itemize}
This explains our interest in BC auctions.

\subsubsection{A useful lemma}
\label{subsec:two}

In what follows the following observation will be helpful.

\begin{lemma}\label{lem:ab}
In each Groves auction for all $\theta \in \Theta$
\begin{enumerate} \smallromani
\item if $\theta_i > \max_{j \neq i} \theta_j$ and $\textrm{argsmax}
  \: (\theta'_i, \theta_{-i}) \neq i$, then $\theta_i + t_i(\theta) >
  t_i(\theta'_i, \theta_{-i})$,

\item  if  $\theta'_i > \max_{j \neq i} \theta_j > \theta_i$, then
$t_i(\theta) >  \theta_i + t_i(\theta'_i, \theta_{-i})$.
\end{enumerate}
  
\end{lemma}

Part $(i)$ states that if player $i$ is a clear winner, given
$\theta_{-i}$ ($\theta_i > \max_{j \neq i} \theta_j$), then it is strictly better
for him to submit his true bid, $\theta_i$, than a losing bid
$\theta'_i$.  In turn, part $(ii)$ states that if player $i$ is a
clear loser, given $\theta_{-i}$ ($\max_{j \neq i} \theta_j > \theta_i$), then it
is strictly better for him to submit his true bid, $\theta_i$, than a
strictly winning bid $\theta'_i$.
\II

\begin{proof}
  Both properties clearly hold for Vickrey auction.  In an arbitrary
  Groves auction the tax for player $i$ is defined by $t_i(\theta) :=
  t^{p}_i(\theta) + r_i(\theta_{-i})$, where $t^{p}_{i}(\theta)$ is
  his tax in Vickrey auction and $(r_1, \LL, r_n)$ is a redistribution
  function that depends only on $\theta_{-i}$. So both properties extend to an arbitrary Groves auction.
\end{proof}

The above lemma will allow us to establish in Section
\ref{sec:seq-groves} some general results about sequential Groves
auctions, which we will later apply to the sequential Vickrey and
sequential BC auctions.

\section{Sequential mechanisms}
\label{sec:sequential1}

We are interested in sequential mechanisms, in particular sequential
auction mechanisms, in which the players announce their types
sequentially.  In this section we review the relevant concepts, some
of which were introduced in \cite{AE08}.

As before, we assume a finite set of decisions $D$, a
set $\{1, \LL, n\}$ of players where $n \geq 2$, and for each player
$i$ a set of types $\Theta_i$ and utility function $ v_i : D \times
\Theta_i \myra \mathbb{R}$.  For notational simplicity, and without
loss of generality, we assume the order to be $1, \LL, n$.

So each player $i$ \emph{knows} the types announced by players $1, \LL, i-1$,
and can use this information to decide which type to announce.
To properly describe this situation we need to specify what is a strategy
in this setting.  

A \oldbfe{strategy} of player $i$ in a sequential mechanism
is a function
\[
s_i: \Theta_1 \times \LL \times \Theta_i \myra \Theta_i.
\]
In this context truth-telling, as a strategy, is
represented by the projection function $\pi_{i}(\cdot)$, defined by
$\pi_i(\theta_1,\LL,,\theta_i) :=\theta_{i}$.  

We assume that in the considered sequential mechanism each
player uses a strategy $s_i(\cdot)$ to select the type he will
announce.  Then if the vector of types that the players
receive is $\theta$ and the vector of strategies that they decide to follow is $s(\cdot) := (s_1(\cdot), \LL, s_n(\cdot))$, the vector of the announced types will be denoted by $[s(\cdot),
\theta]$, where $[s(\cdot), \theta]$ is defined
inductively by $[s(\cdot), \theta]_1 := s_{1}(\theta_1)$ and
$[s(\cdot), \theta]_{i+1} := s_{i+1}([s(\cdot), \theta]_1, \LL,
[s(\cdot), \theta]_i, \theta_{i+1})$.

In what follows we define several properties of strategies that
are appropriate for our analysis of sequential mechanisms.
To start with, we say that strategy $s_i(\cdot)$ of player $i$ is
\oldbfe{dominant} in the sequential version of the mechanism $(f,t)$
if for all
$\theta \in \Theta$, all strategies $s'_i(\cdot)$ of player $i$
and all vectors $s_{-i}(\cdot)$ of strategies of players $j \neq i$
\[
u_{i}((f,t)([(s_{i}(\cdot), s_{-i}(\cdot)), \theta]), \theta_i) \geq
u_{i}((f,t)([(s'_{i}(\cdot), s_{-i}(\cdot)), \theta]), \theta_i).
\]

A weaker notion is that of a rational strategy. We define it by
backward induction. Hence starting with player $n$, we say that strategy
$s_n(\cdot)$ is \oldbfe{rational} in the sequential
version of the mechanism $(f,t)$ if for all $\theta \in \Theta$ and
$\theta'_n \in \Theta_n$
\[
u_n((f,t)(s_n(\theta_1,\LL, \theta_{n}), \theta_{-n}), \theta_n) \geq
u_n((f,t)(\theta'_n, \theta_{-n}), \theta_n).
\]

Assume now that the notion of a rational strategy has been defined for players $n,n-1, \LL, i+1$.
If for any $j \in \C{i+1, \LL, n}$ no rational strategy exists for player $j$, then 
so is the case for player $i$.

Otherwise we say that strategy $s_i(\cdot)$ of player $i$ is
\oldbfe{rational} in the sequential version of the mechanism $(f,t)$ if for all
strategies $s'_{i}(\cdot)$ of player $i$, all
sequences of rational strategies $s_{i+1}(\cdot), \LL, s_{n}(\cdot)$ for players $i+1, \LL, n$,
and all $\theta \in \Theta$
\[
u_i((f,t)([s(\cdot),  \theta], \theta_i) \geq
u_i((f,t)([s'(\cdot), \theta], \theta_i),
\]
where 

\NI
$s(\cdot) := (\pi_1(\cdot), \LL, \pi_{i-1}(\cdot), s_{i}(\cdot),  s_{i+1}(\cdot), 
\LL, s_{n}(\cdot))$,

\NI
$s'(\cdot) := (\pi_1(\cdot), \LL, \pi_{i-1}(\cdot), s'_{i}(\cdot), s_{i+1}(\cdot), 
\LL, s_{n}(\cdot))$.
\II

\NI
(Each strategy $\pi_j(\cdot)$ can be replaced here by an arbitrary
strategy $s_j(\cdot)$.)  Note that for player $n$ the notions of
dominant and rational strategies coincide. However, for player $i$,
where $i < n$ it only holds that every dominant strategy is rational,
provided the set of rational strategies of player $i+1$ (and thus of
players $i+1, \LL, n$) is non-empty.

Another weaker notion is that of an optimal strategy. We say that
strategy $s_i(\cdot)$ of player $i$ is \oldbfe{optimal} in the
sequential version of the mechanism $(f,t)$ if for all $\theta \in
\Theta$ and all $\theta'_i \in \Theta_i$
\[
u_i((f,t)(s_i(\theta_1,\LL, \theta_i), \theta_{-i}), \theta_i) \geq
u_i((f,t)(\theta'_i, \theta_{-i}), \theta_i).
\]
Here as before, $\theta_i$ is the type that player $i$ has received and
$\theta_{-i}$ is the vector of types announced by the other players.

Call a strategy of player $j$ \oldbfe{myopic} if it does not depend on
the types of players $1, \LL, j-1$. Then a strategy $s_i(\cdot)$ of
player $i$ is optimal if for all $\theta \in \Theta$ it yields a best
response to all joint strategies of players $j \neq i$ in which the
strategies of players $i+1, \LL, n$ are myopic. In particular, an
optimal strategy is a best response to the truth-telling by players $j
\neq i$.
By choosing truth-telling as the
strategies of players $j \neq i$ we see that each dominant strategy is
optimal. For player $n$ the concepts of dominant and
optimal strategies coincide.

A particular case of sequential mechanisms are sequential
Groves mechanisms. The following lemma, see \cite{AE08}, provides us
with a sufficient condition for checking whether a strategy is optimal in
such a mechanism.

\begin{lemma} \label{lem:vcg1}
Consider a Groves mechanism $(f,t)$.
Suppose that $s_i(\cdot)$ is a strategy for player $i$ 
such that for all $\theta \in \Theta$,
$f(s_i(\theta_1,\LL,\theta_i), \theta_{-i})=f(\theta)$.
Then $s_i(\cdot)$ is optimal in the sequential version of $(f,t)$.
\end{lemma}
In particular, the strategy $\pi_i(\cdot)$ is optimal in the sequential version of each Groves mechanism.

There are two natural ways of maximizing players' utilitities.
The first one calls for a simultaneous maximization of other players' utilities.
That is, we say that strategy $s_i(\cdot)$ of player $i$ is
\oldbfe{socially maximal} in the sequential version of the mechanism $(f,t)$ if
it is optimal and for all optimal strategies $s'_i(\cdot)$ of player
$i$, all $\theta \in \Theta$ and all $j \neq i$

$u_j((f,t)(s_i(\theta_1,\LL, \theta_i), \theta_{-i}), \theta_j) \geq$

$u_j((f,t)(s'_i(\theta_1,\LL, \theta_i), \theta_{-i}), \theta_j)$.

So a socially maximal strategy of player $i$ simultaneously guarantees
the maximal utility to every other player, under the assumption that
players $i+1, \LL, n$ use myopic strategies (so for instance, truth-telling
strategies).

The second option is to maximize the social welfare.
We say that strategy $s_i(\cdot)$ of player $i$ is
\oldbfe{socially optimal} in the sequential version of the mechanism $(f,t)$ if
it is optimal and for all optimal strategies $s'_i(\cdot)$ of player $i$
and all $\theta \in \Theta$ 

$\sum_{j=1}^{n} u_j((f,t)(s_i(\theta_1,\LL, \theta_i), \theta_{-i}), \theta_j) \geq$ 

$\sum_{j=1}^{n} u_j((f,t)(s'_i(\theta_1,\LL, \theta_i), \theta_{-i}), \theta_j)$.

Hence a socially optimal strategy of player $i$ yields the maximal
social welfare among all optimal strategies, under the assumption that
players $i+1, \LL, n$ use myopic strategies.
Note that each socially maximal strategy is socially optimal. The converse,
as shown in \cite{AE08}, does not hold.
 
Consider now a sequential version of a given mechanism $(f,t)$ and
assume that each player $i$ receives a type $\theta_i \in \Theta_i$
and follows a strategy $s_i(\cdot)$. The resulting social welfare is
\[
SW(\theta, s(\cdot)) := \sum_{j=1}^{n} u_j((f,t)([s(\cdot), \theta]), \theta_j).
\]

We are interested in finding a sequence of optimal players' strategies
for which the resulting social welfare is always maximal. In the subsequent
sections we shall see that such a sequence of strategies can be found
for two natural sequential auction mechanisms.  However, in general, only the
following limited observation can be made.

\begin{lemma} \label{lem:optn}
  Consider a mechanism $(f,t)$ and let $s_n(\cdot)$ be a socially
  optimal strategy for player $n$.  Then
\[
SW(\theta, (s'_{-n}(\cdot), s_n(\cdot))) \geq SW(\theta, s'(\cdot))
\]
for all $\theta$ and all vectors $s'(\cdot)$ of optimal players' strategies.

\end{lemma}

\begin{proof}
The proof follows by the definition of a socially optimal strategy.  
\end{proof}

\section{Sequential Groves auctions}
\label{sec:seq-groves}

In the remainder of the paper we study optimal player strategies in the
sequential Vickrey and BC auctions.
We collect here auxiliary results that hold for all Groves auctions, which we will use in the subsequent two sections.

We shall often rely on the following lemma
concerning optimal strategies.  We stipulate here and elsewhere that for $i = 1$ we have
$\max_{j \in \{1, \ldots, i-1\}} \theta_j = -1$ so that for $i = 1$ we
have $\theta_i > \max_{j \in \{1, \ldots, i-1\}} \theta_j$.

\begin{lemma} \label{lem:vic}

Assume that $s_i(\cdot)$ is an optimal strategy for player $i$ in a sequential Groves auction and that $\theta_1,...,\theta_{i-1}$ are the types announced by the players preceding $i$.

\begin{enumerate} \smallromani
\item Suppose  $\theta_i > \max_{j \in \{1, \ldots, i-1\}} \theta_j$ and $i < n$. Then $s_i(\theta_1, \LL, \theta_i) = \theta_i$.

\item Suppose  $\theta_i > \max_{j \in \{1, \ldots, i-1\}} \theta_j$ and $i = n$. Then $s_i(\theta_1, \LL, \theta_i) >
\max_{j \in \{1, \ldots, i-1\}} \theta_j$.

\item Suppose  $\theta_i \leq \max_{j \in \{1, \ldots, i-1\}} \theta_j$ and $i < n$. Then $s_i(\theta_1, \LL, \theta_i) \leq
\max_{j \in \{1, \ldots, i-1\}} \theta_j$.

\item Suppose  $\theta_i < \max_{j \in \{1, \ldots, i-1\}} \theta_j$ and $i = n$. Then $s_i(\theta_1, \LL, \theta_i) \leq
\max_{j \in \{1, \ldots, i-1\}} \theta_j$.
\end{enumerate}
\end{lemma}

\begin{proof}
  
\NI
$(i)$
Suppose otherwise. If $s_i(\theta_1, \LL, \theta_i) < \theta_i$, then take 
$\epsilon > 0$ such that
$s_i(\theta_1, \LL, \theta_i) + \epsilon < \theta_i$ and set
\[
\theta_{i+1} := \LL := \theta_n := s_i(\theta_1, \LL, \theta_i) + \epsilon.
\]
Then $\theta_i > \max_{j \neq i} \theta_j > s_i(\theta_1, \LL, \theta_i)$, so by Lemma \ref{lem:ab}$(i)$
\begin{align*}
&\phantom{> \ \:}  u_i((f,t)(\theta_i, \theta_{-i}), \theta_i) = \theta_i + t_i(\theta) \\
&>t_i(s_i(\theta_1,\LL, \theta_i), \theta_{-i}) = u_i((f,t)(s_i(\theta_1,\LL, \theta_i), \theta_{-i}), \theta_i).
\end{align*}
This contradicts the optimality of $s_i(\cdot)$.

If $s_i(\theta_1, \LL, \theta_i) > \theta_i$, then take $\epsilon > 0$ such that
$\theta_i + \epsilon < s_i(\theta_1, \LL, \theta_i)$ and set
\[
\theta_{i+1} := \LL := \theta_n := \theta_i + \epsilon.
\]
Then $f(\theta) \neq i$, so
$u_i((f,t)(\theta_i, \theta_{-i}), \theta_i) = r_i(\theta_{-i})$, where
$t_i(\theta) := t^{p}_i(\theta) + r_i(\theta_{-i})$. 
On the other hand $f(s_i(\theta_1,\LL, \theta_i), \theta_{-i}) = i$ and consequently
\[
u_i((f,t)(s_i(\theta_1,\LL, \theta_i), \theta_{-i}), \theta_i) = \theta_i - (\theta_i + \epsilon) + r_i(\theta_{-i})
<  r_i(\theta_{-i}).
\]
This again contradicts the optimality of $s_i(\cdot)$.
\II

\NI
$(ii)$
Suppose otherwise, that is
$s_n(\theta_1, \LL, \theta_n) \leq \max_{j \neq n} \theta_j$.
Then \\
$\textrm{argsmax} \: (s_n(\theta_1,\LL, \theta_n), \theta_{-n}) \neq n$, so 
by Lemma \ref{lem:ab}$(i)$
\begin{align*}
&\phantom{> \ \:}  u_n((f,t)(\theta_n, \theta_{-n}), \theta_n) = \theta_n + t_n(\theta) \\
&>t_n(s_n(\theta_1,\LL, \theta_n), \theta_{-n}) = u_n((f,t)(s_n(\theta_1,\LL, \theta_n), \theta_{-n}), \theta_n).
\end{align*}
This contradicts the optimality of $s_n(\cdot)$.
\II

\NI
$(iii)$
Suppose otherwise, i.e., $s_i(\theta_1, \LL, \theta_i) > \max_{j \in
  \{1, \LL, i-1\}} \theta_j$.  Then take $\epsilon > 0$ such that
$\max_{j \in \{1, \LL, i-1\}} \theta_j + \epsilon < s_i(\theta_1, \LL, \theta_i)$ and set
\[
\theta_{i+1} := \LL := \theta_n := \max_{j \in \{1, \LL, i-1\}} \theta_j + \epsilon.
\]
Then $s_i(\theta_1, \LL, \theta_i) > \max_{j \neq i} \theta_j > \theta_i$, so by Lemma \ref{lem:ab}$(ii)$
\begin{align*}
&\phantom{> \ \:}  u_i((f,t)(\theta_i, \theta_{-i}), \theta_i) = \theta_i \\
&> \theta_i + t_i(s_i(\theta_1,\LL, \theta_i), \theta_{-i}) = u_i((f,t)(s_i(\theta_1,\LL, \theta_i), \theta_{-i}), \theta_i).
\end{align*}
This contradicts the optimality of $s_i(\cdot)$.
\II

\NI
$(iv)$
Suppose otherwise, that is
$s_n(\theta_1, \LL, \theta_n) > \max_{j \neq n} \theta_j$.
Then \\
$\textrm{argsmax} \: (s_n(\theta_1,\LL, \theta_n), \theta_{-n}) = n$, so 
by Lemma \ref{lem:ab}$(ii)$

$u_n((f,t)(\theta_n, \theta_{-n}), \theta_n) = t_n(\theta) >\theta_n + t_n(s_n(\theta_1,\LL, \theta_n), \theta_{-n}) =$

$u_n((f,t)(s_n(\theta_1,\LL, \theta_n), \theta_{-n}), \theta_n)$.

This contradicts the optimality of $s_n(\cdot)$.
\end{proof}

This allows us to draw some helpful conclusions.

\begin{corollary} \label{cor:f}
In each sequential Groves auction for all $\theta \in \Theta$ and all vectors $s(\cdot)$ of optimal players' strategies

\begin{enumerate} \smallromani

\item for all $i \in \{1, \LL, n-1\}$, $\max_{j \in \{1, \ldots, i\}} [s(\cdot), \theta]_j =$ \\
$\max_{j \in \{1, \ldots, i\}} \theta_j$,

\item for all $i \in \{1, \LL, n-1\}$, ${\rm argsmax}_{j \in \{1, \ldots, i\}} [s(\cdot), \theta]_j = {\rm argsmax}_{j \in \{1, \ldots, i\}} \theta_j$,
  
\item for all $i \in \{1, \LL, n-1\}$, if 
  $\theta_{i} > \max_{j\in\{1,\ldots,i-1\}}
  [s(\cdot),\theta]_j$, then $\theta_{i} > \max_{j\in\{1,\ldots,i-1\}}
  \theta_j$ and also for any other vector of optimal players' strategies
  $s'(\cdot)$ we have $\theta_{i} > \max_{j\in\{1,\ldots,i-1\}}
  [s'(\cdot),\theta]_j$,

\item either $f([s(\cdot), \theta]) = f(\theta)$ or if not, 
$\theta_n = \max_{i \neq n} \theta_i$, 
$[s(\cdot), \theta]_n > \theta_n$ and $f([s(\cdot), \theta]) = n$.
\end{enumerate}
\end{corollary}

Informally, items $(i)-(iii)$ state that when each player follows an optimal strategy,
the first $n-1$ of entries of $\theta$ and $[s(\cdot), \theta]$ are very similar.
In turn, item $(iv)$ states that, except when $\theta_n = \max_{i
  \neq n} \theta_i$ and player $n$ submits a larger bid, the same
outcome is realized under an arbitrary vector of optimal strategies as
under truth-telling.  This exception does not take place for the
specific optimal strategies we consider in the sequel.

\begin{proof}

\NI
$(i)$ and $(ii)$ follow by a straightforward induction using Lemma \ref{lem:vic}.
$(iii)$ is a direct consequence of $(i)$.
 To prove $(iv)$ note that if $\theta_n = \max_{i \neq n} \theta_i$ and
 $f([s(\cdot), \theta]) \neq f(\theta)$, 
 then by $(ii)$
 $[s(\cdot), \theta]_n > \theta_n$ and hence $f([s(\cdot), \theta]) = n$.
 Otherwise $\theta_n \neq \max_{j \in \{1, \ldots, n-1\}} \theta_j$ 
 and two cases arise.
 \II

 \NI
 \emph{Case 1} $\theta_n > \max_{j \in \{1, \ldots, n-1\}} \theta_j$. 

 Then by $(i)$ we have $\theta_n > \max_{j \in \{1, \ldots, n-1\}}
 [s(\theta), \theta]_j$, so by Lemma \ref{lem:vic}$(ii)$ $[s(\cdot), \theta]_n > \max_{j \in \{1, \ldots, n-1\}} [s(\cdot), \theta]_j$.
 Hence by $(i)$ and $(ii)$ we get $\textrm{argsmax} \:
 [s(\cdot), \theta] = n$ and $\textrm{argsmax} \: \theta = n$, that is
 $f([s(\cdot), \theta]) = f(\theta)$.  
 \II

 \NI
 \emph{Case 2} $\theta_n < \max_{j \in \{1, \ldots, n-1\}} \theta_j$.  

 Then by $(i)$  $\theta_n < \max_{j \in \{1, \ldots, n-1\}}
 [s(\theta), \theta]_j$, so by Lemma \ref{lem:vic}$(iv)$ $[s(\cdot),
 \theta]_n \leq \max_{j \in \{1, \ldots, n-1\}} [s(\cdot), \theta]_j$.
 Consequently $\textrm{argsmax} \: [s(\cdot), \theta] =
 \textrm{argsmax}_{j \in \{1, \ldots, n-1\}} [s(\cdot), \theta]_j$.
 Also \\
 $\textrm{argsmax} \: \theta = \textrm{argsmax}_{j \in \{1,
   \ldots, n-1\}} \theta_j$.  So by $(ii)$ we get \\
 $\textrm{argsmax} \: [s(\cdot), \theta] = \textrm{argsmax} \: \theta$,
 that is $f([s(\cdot), \theta]) = f(\theta)$. 
 \end{proof}

\section{Sequential Vickrey auctions}
\label{sec:vic}

We now focus on sequential Vickrey auctions. 
We shall need the following observation.

\begin{lemma} \label{lem:max}
Consider a sequential Vickrey auction.
For all $\theta \in \Theta$, all vectors $s(\cdot)$ of optimal players' strategies,
if $\theta_n = \max_{i \in \{1, \ldots, n-1\}} \theta_i$, then 
$u_n((f,t)([s(\cdot), \theta], \theta_n) = 0$.
\end{lemma}

\begin{proof}
It suffices to consider the case when $n = f([s(\cdot), \theta])$.
Then $[s(\cdot), \theta]_n > \max_{j \neq n} [s(\cdot), \theta]_j$, so
by Corollary \ref{cor:f}$(i)$ $[s(\cdot), \theta]^*_2 = \max_{i \neq n} \theta_i = \theta_n$.
Hence
$u_n((f,t)([s(\cdot), \theta], \theta_i) = 0$.
\end{proof}

First, we establish the following negative result.

\begin{theorem} \label{thm:dom}
Consider a sequential Vickrey auction.

\begin{enumerate} \smallromani
\item For $i \in \{1, \LL, n-1\}$ no dominant strategy exists 
for player $i$.

\item Every strategy $s_n(\cdot)$ such that
\[
        \begin{array}{l@{\extracolsep{3mm}}l}
        s_n(\theta_1, \LL, \theta_n) > \max_{j \neq n} \theta_j   & 
\mathrm{if}\  \theta_n > \max_{j \neq n} \theta_j, \\
        s_n(\theta_1, \LL, \theta_n) \leq \max_{j \neq n} \theta_j & \mathrm{otherwise}
        \end{array}
\]
is dominant (and hence rational) for player $n$.

\item For $i \in \{1, \LL, n-1\}$ no rational strategy exists 
for player $i$.

\end{enumerate}

\end{theorem}

\begin{proof}

\NI
$(i)$
Suppose otherwise. Let $s_i(\cdot)$ be a dominant strategy of player $i$. 
In particular $s_i(\cdot)$ is optimal. Choose now $\theta \in \Theta$ such that
$\theta_i = 2$ and $\theta_j = 0$ for $j \neq i$.
By Lemma \ref{lem:vic}$(i)$ $s_i(\theta_1, \LL, \theta_i) = \theta_i$.
Take now truth-telling as  strategy for players $j < i$ and
the following strategy for players $j > i$:

$s_j(\theta_1, \LL, \theta_j) :=$
\[
 \left\{
        \begin{array}{l@{\extracolsep{3mm}}l}
        \theta_j    & \mathrm{if}\  \theta_j > \max_{k \in \{1, \ldots, j-1\}} \theta_k, \\
        \max \{\max_{k \in \{1, \ldots, j-1\}} \theta_k- 1, 0\}  & \mathrm{otherwise}.
        \end{array}
        \right.
\]

Then 
$u_{i}((f,t)([(s_{i}(\cdot), s_{-i}(\cdot)), \theta]), \theta_i) = 1$,
while strategy $s_i(\cdot)$ such that 
$s_i(\theta_1, \LL, \theta_i) = 1$ yields player $i$ in this case the final utility $2$.
\II

\NI
$(ii)$
It suffices to recall that for player $n$ the concepts of
dominant, rational and optimal strategies coincide and apply Lemma \ref{lem:vcg1}.
\II

\NI
$(iii)$
It suffices to prove the claim for $i = n-1$.
Suppose by contradiction that a rational strategy $s_{n-1}(\cdot)$ of
player $n-1$ exists.  By $(ii)$ the truth-telling strategy
$\pi_n(\cdot)$ is a rational strategy for player $n$.  By the
definition of a rational strategy, using $s'_{n-1}(\cdot) =
\pi_{n-1}(\cdot)$ and $s_n(\cdot) = \pi_n(\cdot)$, we have for all
$\theta \in \Theta$
\begin{align*}
&\phantom{\geq \ \:}  u_{n-1}((f,t)(s_{n-1}(\theta_1,\LL, \theta_{n-1}), \theta_{-{(n-1)}}), \theta_{n-1}) \\
&\geq u_{n-1}((f,t)(\theta), \theta_{n-1}).
\end{align*}

By Lemma \ref{lem:vcg1} $\pi_{n-1}(\cdot)$ is an optimal strategy for
player $n-1$, so we have for all $\theta \in \Theta$ and all $\theta'_i \in
\Theta_i$
\[
u_{n-1}((f,t)(\theta), \theta_{n-1}) \geq
u_{n-1}((f,t)(\theta'_{n-1}, \theta_{-{(n-1)}}), \theta_{n-1}).
\]

The above two inequalities imply that $s_{n-1}(\cdot)$ is optimal.
It suffices now to take $i = n-1$ and repeat the proof of $(i)$ and note that strategy 
$s_n(\cdot)$ defined there is rational for player $n$.
\end{proof}

So we shall focus on the weaker notion of optimal strategy.
The following natural strategy for player $i$ is an example of an optimal strategy
that deviates from truth-telling:
\begin{equation}
\label{equ:opt}
s_i(\theta_1, \LL, \theta_i) :=
 \left\{
        \begin{array}{l@{\extracolsep{3mm}}l}
        \theta_i    & \mathrm{if}\  \theta_i > \max_{j \in \{1, \ldots, i-1\}} \theta_j, \\
        0  & \mathrm{otherwise}.
        \end{array}
        \right.
\end{equation}

Note that strategy $s_i(\cdot)$ is indeed optimal, since it does not change
the decision that would be taken if player $i$ is truthful, and hence
Lemma~\ref{lem:vcg1} can be applied.
When we limit our attention to optimal strategies we get the following result.

\begin{theorem}
\label{thm:dom-opt}
Consider a sequential Vickrey auction.
For all $\theta \in \Theta$, all vectors $s(\cdot)$ of optimal players' strategies
and all optimal strategies $s'_i(\cdot)$ of player $i$
\[
u_i((f,t)([s(\cdot), \theta], \theta_i) =
u_i((f,t)([(s'_i(\cdot), s_{-i}(\cdot)), \theta], \theta_i).
\]
\end{theorem}

This can be interpreted as a statement that each optimal strategy is dominant within the
universe of optimal strategies. 

\begin{proof}
\II

\NI
\emph{Case 1} $f(\theta)=i$. 

If $i < n$, then by Lemma \ref{lem:vic}$(i)$ and Corollary \ref{cor:f}$(i)$
we have
$
[s(\cdot), \theta]_i = \theta_i = [(s'_i(\cdot), s_{-i}(\cdot)), \theta]_i,
$
so $[s(\cdot), \theta] = [(s'_i(\cdot), s_{-i}(\cdot)), \theta]$
and the desired equality holds.

If $i = n$, then for $j = 1, \LL, n-1$ we have
$[s(\cdot), \theta]_j = [(s'_n(\cdot), s_{-n}(\cdot)), \theta]_j$.
Hence
$[s(\cdot), \theta]^*_2 = [(s'_n(\cdot), s_{-n}(\cdot)), \theta]^*_2$.
Additionally $\theta_n \neq \max_{i \in \{1, \ldots, n-1\}} \theta_i$ (otherwise $f(\theta) \neq n$),
so by Corollary \ref{cor:f}$(iv)$
$f([s(\cdot), \theta]) = f([(s'_i(\cdot), s_{-i}(\cdot)), \theta]) = i$ and consequently
\begin{align*}
&\phantom{= \ \:}  u_i((f,t)([s(\cdot), \theta], \theta_i) = \theta_i - [s(\cdot), \theta]^*_2 \\
&=\theta_i - [(s'_i(\cdot), s_{-i}(\cdot)), \theta]^*_2 = u_i((f,t)([(s'_i(\cdot), s_{-i}(\cdot)), \theta], \theta_i).
\end{align*}

\II

\NI
\emph{Case 2} $f(\theta) \neq i$. 

If both $f([s(\cdot), \theta]) \neq i$ and $f([(s'_i(\cdot),
s_{-i}(\cdot)), \theta]) \neq i$, then both $u_i((f,t)([s(\cdot),
\theta], \theta_i) = 0$ and $u_i((f,t)([(s'_i(\cdot), s_{-i}(\cdot)),
\theta], \theta_i) = 0$.  

Otherwise, $f([s(\cdot), \theta]) \neq f(\theta)$ or $f([(s'_i(\cdot),
s_{-i}(\cdot)), \theta]) \neq f(\theta)$, so by Corollary \ref{cor:f}$(iv)$
$\theta_n = \max_{i \in \{1, \ldots, n-1\}} \theta_i$. 

If $i < n$, by symmetry, the only case to consider
is when $f([s(\cdot), \theta]) = i$ and $f([(s'_i(\cdot), s_{-i}(\cdot)), \theta]) \neq i$.
Then $f([s(\cdot), \theta]) \neq f(\theta)$, so $f([s(\cdot), \theta]) = n$, so this case cannot arise.
If $i = n$, the desired equality holds by Lemma \ref{lem:max}.
\end{proof}

This shows that from the point of view of each player all optimal
strategies are equivalent (assuming that each player has at his
disposal only optimal strategies). However, optimal strategies may
differ when the players take into account the utility of other
players, in particular, the social welfare.  In \cite{AE08} it was
proved that in the well-known case of the public project problem (see,
e.g., \cite[page 861]{MWG95}) socially maximal strategies do not
exist. (It was also showed there that socially optimal strategies do
exist.)  The following result shows that in the case of the sequential
Vickrey auctions the situation changes and that strategy $s_i(\cdot)$
defined in (\ref{equ:opt}) plays then a special role.

\begin{theorem} \label{thm:max}
In the sequential Vickrey auction
strategy $s_i(\cdot)$ defined in (\ref{equ:opt}) is socially maximal for player $i$.
\end{theorem}

\begin{proof}
We noted already that by virtue of Lemma \ref{lem:vcg1} strategy $s_i(\cdot)$ defined by (\ref{equ:opt}) is optimal.
We also need to prove that for all optimal strategies $s'_i(\cdot)$ of player $i$, all 
$\theta \in \Theta$  and all $j \neq i$
\begin{equation}
  \label{equ:max}
u_j((f,t)(s_i(\theta_1,\LL, \theta_i), \theta_{-i}), \theta_j) \geq \\
u_j((f,t)(s'_i(\theta_1,\LL, \theta_i), \theta_{-i}), \theta_j).
\end{equation}

Fix $\theta \in \Theta$ and $j \neq i$.
Consider Corollary \ref{cor:f}$(iv)$ with $s_j(\cdot) := \pi_j(\cdot)$ for $j \neq i$. Then either
\[
f(s_i(\theta_1,\LL, \theta_i), \theta_{-i}) = f(s'_i(\theta_1,\LL, \theta_i), \theta_{-i}) = f(\theta)
\]
or $\theta_n = \max_{i \neq n} \theta_i$.

Consider first the former case.
If $f(\theta) \neq j$, then
\begin{align*}
&\phantom{= \ \:}  u_j((f,t)(s_i(\theta_1,\LL, \theta_i), \theta_{-i}), \theta_j) \\
&= u_j((f,t)(s'_i(\theta_1,\LL, \theta_i), \theta_{-i}), \theta_j) = 0.
\end{align*}
Otherwise 
\[
u_j((f,t)(s_i(\theta_1,\LL, \theta_i), \theta_{-i}), \theta_j) = \theta_j - {(s_i(\theta_1,\LL, \theta_i), \theta_{-i})}^*_{2}
\]
and
\[
u_j((f,t)(s'_i(\theta_1,\LL, \theta_i), \theta_{-i}), \theta_j) = \theta_j - {(s'_i(\theta_1,\LL, \theta_i), \theta_{-i})}^*_{2}.
\]
Since $f(\theta) \neq j$, either contingency $(i)$ or contingencies
$(iii)$ or $(iv)$ of Lemma \ref{lem:vic} apply.  In the first case
$s_i(\theta_1, \LL, \theta_i) = s'_i(\theta_1, \LL, \theta_i) =
\theta_i$, so (\ref{equ:max}) holds.  In the second case
$s_i(\theta_1,\LL, \theta_i) = 0$ and for all $\theta'_i \geq 0$ we
have $- {(0, \theta_{-i})}^*_{2} \geq - {(\theta'_i, \theta_{-i})}^*_{2}$,
so (\ref{equ:max}) holds, as well.

In the latter case note that by Corollary \ref{cor:f}$(iv)$ we have
$f(s_i(\theta_1,\LL, \theta_i), \theta_{-i}) = f(\theta)$ since $[(s_i(\cdot), (\pi_{-i}(\cdot))), \theta]_n = 0 \leq \theta_n$.
So $f(s'_i(\theta_1,\LL, \theta_i), \theta_{-i}) \neq f(\theta)$ and
$[(s'_i(\cdot), (\pi_{-i}(\cdot))), \theta]_n > \theta_n$.
Hence $f(s'_i(\theta_1,\LL, \theta_i), \theta_{-i}) = n$ and by 
Lemma \ref{lem:max} for all $j \in \{1, \LL, n\}$ we have
$u_j((f,t)(s'_i(\theta_1,\LL, \theta_i), \theta_{-i}), \theta_j) = 0$, so 
(\ref{equ:max}) holds.
\end{proof}

Finally, we show that when each player $i$ follows strategy $s_i(\cdot)$
of Theorem \ref{thm:max}, maximal social welfare 
is generated.

\begin{theorem} \label{thm:vic-opt}
In the sequential Vickrey auction
for all $\theta \in \Theta$ and vectors $s'(\cdot)$ 
of optimal players' strategies,
\[
SW(\theta, s(\cdot)) \geq SW(\theta, s'(\cdot))
\]
where $s(\cdot)$ is the vector of strategies $s_i(\cdot)$ defined in (\ref{equ:opt}).
\end{theorem}

\begin{proof}
Fix $\theta \in \Theta$. By Lemma \ref{lem:optn} it suffices to prove that
for all vectors $s'_{-n}(\cdot)$ of socially optimal strategies for players $1, \LL, n-1$
\[
SW(\theta, s(\cdot)) \geq SW(\theta, (s'_{-n}(\cdot), s_n(\cdot))),
\]
where $s(\cdot)$ is the vector of strategies $s_i(\cdot)$ of Theorem \ref{thm:max}. 

Fix a vector $s'_{-n}(\cdot)$ of socially optimal strategies for players $1, \LL, n-1$.
By Corollary \ref{cor:f}$(iv)$ either for some $j$ we have
$f([s(\cdot), \theta]) = f([(s'_{-n}(\cdot), s_n(\cdot)), \theta]) = j$ or
$\theta_n = \max_{i \in \{1, \ldots, n-1\}} \theta_i$.

In the former case
\[
SW(\theta, s(\cdot)) = u_j((f,t)([s(\cdot), \theta]), \theta_j) = \theta_j - {[s(\cdot), \theta]}^*_2
\]
and
\begin{align*}
&\phantom{= \ \:}  SW(\theta, (s'_{-n}(\cdot), s_n(\cdot))) = u_j((f,t)([(s'_{-n}(\cdot), s_n(\cdot)), \theta]), \theta_j) \\
&= \theta_j - {[(s'_{-n}(\cdot), s_n(\cdot)), \theta]}^*_2.
\end{align*}

Further, a straightforward proof by induction using Corollary \ref{cor:f}$(i)$ and
Lemma \ref{lem:vic}$(i)$ and $(iii)$ shows that for all $i \in \{1, \LL, n-1\}$, 
$[s(\cdot), \theta]_i \leq [(s'_{-n}(\cdot), s_n(\cdot)), \theta]_i$.
Hence ${[s(\cdot), \theta]}^*_2 \leq {[(s'_{-n}(\cdot), s_n(\cdot)), \theta]}^*_2$, so 

$SW(\theta, s(\cdot)) \geq SW(\theta, (s'_{-n}(\cdot), s_n(\cdot)))$.

In the latter case $f([s(\cdot), \theta]) = f(\theta)$ so $f([(s'_{-n}(\cdot), s_n(\cdot)), \theta]) \neq f(\theta)$
and hence by Corollary \ref{cor:f}$(i)$ $f([(s'_{-n}(\cdot), s_n(\cdot)), \theta]) = n$, so
by Lemma \ref{lem:max}
\[
SW(\theta, (s'_{-n}(\cdot), s_n(\cdot))) = u_n((f,t)([(s'_{-n}(\cdot), s_{n}(\cdot)), \theta], \theta_n) =0,
\]
so the desired inequality holds.
\end{proof}

This maximal final social welfare under $s(\cdot)$
equals 
\[
SW(\theta, s(\cdot)) = \theta_i - \max_{i \in \{1, \ldots, i-1\}} \theta_i,
\] 
where $i = \textrm{argsmax} \: \theta$. This is always greater than or equal to the final
social welfare achieved in a Vickrey auction when players bid truthfully, which is $\theta_i -
\max_{i \neq j} \theta_i$.
Additionally, for some inputs, for
instance those of the form $\theta^{*}$, with the first three entries
different, it is strictly greater.

\section{Sequential BC auctions}
\label{sec:bc}

Next, we consider sequential Bailey-Cavallo auctions.  We first show
that in analogy to the sequential Vickrey auctions no dominant
strategies exist except for the last player. In fact we establish this
for a wide class of Groves auctions. 

\begin{theorem}\label{thm:bcdom}
  Consider a sequential Groves auction such that the redistribution
  function $r$ satisfies: $0\leq r_i(\theta_{-i}) <
  {(\theta_{-i})}^*_1$ for all $i$.
\begin{enumerate} \smallromani
\item For $i \in \{1, \LL, n-1\}$ no dominant strategy exists 
for player $i$.
\item Every strategy $s_n(\cdot)$ such that
\[
        \begin{array}{l@{\extracolsep{3mm}}l}
        s_n(\theta_1, \LL, \theta_n) > \max_{j \neq n} \theta_j   & 
\mathrm{if}\  \theta_n > \max_{j \neq n} \theta_j, \\
        s_n(\theta_1, \LL, \theta_n) \leq \max_{j \neq n} \theta_j & \mathrm{otherwise}
        \end{array}
\]
%
is dominant for player $n$.
\item For $i \in \{1, \LL, n-1\}$ no rational strategy exists 
for player $i$.
\end{enumerate}
\end{theorem}
The proof is based on the same arguments as the proof of Theorem~\ref{thm:dom}.
and is omitted.




We shall thus focus, as in the case of sequential Vickrey auctions, on the weaker notion of optimal strategy.
We have various natural optimal strategies that deviate from truth-telling, such as the following one:
\begin{equation}
\label{dominant2}
s_i(\theta_1,\ldots,\theta_i) := \left\{\begin{array}{ll}
                   \theta_i & \mbox{if } \theta_i > \max_{j\in\{1,\ldots,i-1\}} \theta_j \\
                   {(\theta_1,\ldots,\theta_{i-1})}^{*}_1 & \mbox{if } \theta_i \leq \max_{j\in\{1,\ldots,i-1\}} \theta_j \\
                                                        & \mbox{and } i\leq n-1 \\                   
                   {(\theta_1,\ldots,\theta_{i-1})}^{*}_2 &\mbox{otherwise}
                   \end{array}
             \right. 
\end{equation}

According to strategy $s_i(\cdot)$ if player $i$ cannot be a winner when bidding truthfully 
($\theta_i \leq \max_{j\in\{1,\ldots,i-1\}} \theta_j$)
he submits a bid that equals the highest current bid if $i < n$ 
or the second highest current bid if $i = n$.
Note that strategy $s_i(\cdot)$ is indeed optimal in the sequential BC auction,
since it does not change the decision that would be taken if player $i$ is truthful, so
Lemma~\ref{lem:vcg1} can be applied. We will see later that strategy $s_i(\cdot)$ has some desirable properties regarding the welfare of the players under the BC auction.

We now show that the analogue of Theorem~\ref{thm:dom-opt} does not hold for the sequential BC auctions.
In particular this means that optimal strategies are not dominant within the universe of optimal strategies in the sequential BC auction.

\begin{theorem}
Consider a sequential BC auction.
There exists a type vector $\theta \in \Theta$, a vector of optimal strategies $s(\cdot)$ and an optimal strategy $s_i'(\cdot)$
for some player $i$, such that:
\[
u_i((f,t)([(s'_i(\cdot), s_{-i}(\cdot)), \theta], \theta_i) >
u_i((f,t)([s(\cdot), \theta], \theta_i).
\]
\end{theorem}
\begin{proof}
Consider the following two strategies: 
\[
s_i(\theta_1,\ldots,\theta_i) :=
 \left\{
        \begin{array}{l@{\extracolsep{3mm}}l}
        \theta_i    & \mathrm{if}\  \theta_i > \max_{j \in \{1, \ldots, i-1\}} \theta_j, \\
        \max  \{0,\theta_{i-1}-\epsilon\}  & \mathrm{otherwise}
        \end{array}
        \right.
\]
where $\epsilon >0$, and
\[
s'_i(\theta_1,\ldots,\theta_i) :=
 \left\{
        \begin{array}{l@{\extracolsep{3mm}}l}
        \theta_i    & \mathrm{if}\  \theta_i > \max_{j \in \{1, \ldots, i-1\}} \theta_j, \\
        \theta_{i-1}  & \mathrm{otherwise}
        \end{array}
        \right.
\]
By Lemma~\ref{lem:vcg1} both $s_i(\cdot)$ and
  $s'_i(\cdot)$ are optimal strategies since they do not alter the
  outcome that would be realized if players are thruthful. Consider
  now an instance where $n=3$ and the type vector is $\theta =
  (10,9,8)$. Under the joint strategy $s(\cdot)$ the submitted types are $(10,
  10-\epsilon, 10-2\epsilon)$ and the second player receives a
  redistribution of $\frac{10-2\epsilon}{3}$. However, if second player
  switches to $s_i'(\cdot)$, the submitted types will be $(10, 10, 10-\epsilon)$.
  Hence the redistribution to the second player is now
  $\frac{10-\epsilon}{3}$.
\end{proof}

We now turn to the question of existence of socially optimal
strategies, which we answer negatively. This is again in contrast to the
sequential Vickrey auction for which we established in Theorem \ref{thm:max}
the existence of even socially maximal strategies. 
 
\begin{theorem}
The sequential BC auction does not admit socially optimal strategies except for the first and last player.
\end{theorem}

\begin{proof}

Take $i \in \{2, \LL, n-1\}$. 
Suppose that the announced types $\theta_1,\LL,\theta_{i-1}$ are distinct and
  $\theta_i$ is the second highest bid within
  $\{\theta_1,\ldots,\theta_i\}$. Suppose $s_i(\cdot)$ is an optimal
  strategy for $i$. We show that it cannot be socially optimal in
  the BC auction.  To see this, note first that $\theta_i <
  \max_{j\in\{1,\ldots,i-1\}} \theta_j$. Hence by Lemma~\ref{lem:vic}$(iii)$
  $s_i(\theta_1,\ldots,\theta_i)\leq \max_{j\in\{1,\ldots,i-1\}}
  \theta_j$. We distinguish two cases.

\II

\noindent {\it Case 1} $s_i(\theta_1,\ldots,\theta_i) = \theta_i$. 

Then consider the following completion of the type vector:
$\theta_{i+1}$ is strictly in between $\theta_i$ and
$\max_{j\in\{1,\ldots,i-1\}} \theta_j$ and all the remaining types are
less than $\theta_i$. In this case the second and third highest bids
are $\theta_{i+1}$ and $\theta_i$ respectively and the total sum of
taxes is $\frac{2}{n}(\theta_{i+1} - \theta_i)\neq 0$. However, if
player $i$ had submitted $\theta'_i:= \theta_{i+1}$, which is an
optimal strategy, the sum of taxes would have been $0$.  
\II

\noindent {\it Case 2} $s_i(\theta_1,\ldots,\theta_i) \neq \theta_i$. 

Then $s_i(\theta_1,\ldots,\theta_i) \leq \max_{j\in\{1,\ldots,i-1\}}
\theta_j$, by Lemma~\ref{lem:vic}. Consider now the following
completion of the type vector: $\theta_{i+1} = \theta_i$ and all the
remaining types are lower than $\theta_i$. In this case, under
$s_i(\cdot)$, the second and third highest bids would not coincide and
hence the sum of taxes would not be equal to $0$. On the other hand,
for $\theta'_i:=\theta_i$, the sum of the taxes would be $0$.

Hence we can always find a completion of the type vector and an
optimal strategy $s'_i$ for player $i$ such that the sum of utilities
is higher under $s'_i$ than under $s_i$.
\end{proof}
%

The results established so far show that the sequential Vickrey
auctions and BC auctions differ in many ways.  We conclude by showing
that they do share one property. Namely, within the universe of
optimal strategies there exists an optimal strategy $s_i(\cdot)$ such
that if all players follow it, then maximal social welfare is
generated for all $\theta\in\Theta$. The desired strategy is the one
introduced in (\ref{dominant2}). This is in analogy to Theorem
\ref{thm:vic-opt}. However, the optimal strategy used in
Theorem~\ref{thm:vic-opt} is tailored to the sequential Vickrey
auction and is quite different from the strategy employed here.

\begin{theorem}\label{thm:opt-bc}
In the sequential BC auction
for all $\theta \in \Theta$ and all vectors $s'(\cdot)$ of optimal players' strategies,
\[
SW(\theta, s(\cdot)) \geq SW(\theta, s'(\cdot))
\]
where $s(\cdot)$ is the vector of strategies $s_i(\cdot)$ defined in (\ref{dominant2}). 
\end{theorem}
\begin{proof}
Consider a type vector $\theta = (\theta_1,\LL,\theta_n)$ and let
$s'(\cdot)$ be an arbitrary vector of optimal strategies.
By Corollary~\ref{cor:f}$(iv)$ the initial social welfare, i.e., the (initial) utility of the winner,
is the same under $s(\cdot)$ and $s'(\cdot)$. Hence we only need to compare the
aggregate tax paid by the players.

We now proceed in three steps. First suppose that
$\theta$ is such that $\theta_n \leq \max_{j\in\{1,\ldots,n-1\}}
[s(\cdot),\theta]_j$.
By the definition of $s(\cdot)$ in (\ref{dominant2}) and the definition of
the BC auction (see (\ref{eq:bc-tax})), we know that player $n$ submits the currently second
highest bid under $s(\cdot)$ and as a result the taxes are $0$. Hence for such type vectors
maximal social welfare is generated when players follow $s(\cdot)$.

Thus we may assume that $\theta_n > \max_{j\in\{1,\ldots,n-1\}}
[s(\cdot),\theta]_j$.
Suppose now that $\theta_{n-1} \leq \max_{j\in\{1,\ldots,n-2\}}
[s(\cdot),\theta]_j$. Then player $n-1$ submits the currently highest bid
and since player $n$ has the highest bid, the taxes sum up to $0$ by (\ref{eq:bc-tax}).

Hence the remaining case is when $\theta_n >
\theta_{n-1} >$ \\ $\max_{j\in\{1,\ldots,n-2\}} [s(\cdot),\theta]_j$. In this case we know that when the players follow $s(\cdot)$,
player $n$ is the winner and player $n-1$ submits $\theta_{n-1}$.
We claim that the same happens under any other optimal strategy.

Lemma~\ref{lem:vic} and Corollary \ref{cor:f}$(iii)$ imply that when
players follow $s'(\cdot)$, player $n$ is the winner and player $n-1$
submits $\theta_{n-1}$. Thus the aggregate tax under $s(\cdot)$ is:
$$ \sum_i t_i(\theta) = \frac{2}{n}(\theta_{n-1} - {[s(\cdot),\theta]}^{\ast}_3)$$
and under $s'(\cdot)$:
$$ \sum_i t_i(\theta) = \frac{2}{n}(\theta_{n-1} - {[s'(\cdot),\theta]}^*_3).$$
 
It suffices now to compare the third highest bid in $s(\cdot)$ and
$s'(\cdot)$. 
For this we need the following more general claim.
\begin{claim}\label{thirdhighest}
Let $\theta\in\Theta$, $z = ([s(\cdot),\theta]_1,\ldots,[s(\cdot),\theta]_{n-1})$ and
$z' = ([s'(\cdot),\theta]_1,\ldots$, $[s'(\cdot),\theta]_{n-1})$. Then for all $i\in\{1,\ldots,n-1\}$, $z_i\geq z_i'$.
\end{claim}

This claim also helps in understanding the intuition behind strategy
$s_i(\cdot)$. In particular, under the joint strategy $s(\cdot)$, any
bidder, except the last one, whose type is no more than the currently
highest bid, is bidding the maximal possible value among the set of
his optimal strategies.  
\II

\NI
\emph{Proof.}
The proof is by induction using Lemma~\ref{lem:vic}.  
If $i=1$, then by Lemma~\ref{lem:vic}, $z_1 = z_1' = \theta_1$.
For the induction step, suppose the Lemma holds for all $i<k$, where $k\geq 2$.
We argue about $z_k$ and $z_k'$. By the induction hypothesis $\max_{j\in\{1,\ldots,k-1\}} z_j \geq \max_{j\in\{1,\ldots,k-1\}} z'_j$.\\

\noindent {\it Case 1} $\theta_k > \max_{j\in\{1,\ldots,k-1\}} z_j \geq \max_{j\in\{1,\ldots,k-1\}} z'_j$. Then by Lemma~\ref{lem:vic} $z_k = z'_k = \theta_k$.
\II

\noindent {\it Case 2} $ \max_{j\in\{1,\ldots,k-1\}} z_j \geq \max_{j\in\{1,\ldots,k-1\}} z'_j \geq \theta_k$. Then Lemma~\ref{lem:vic} implies $z_k = \max_{j\in\{1,\ldots,k-1\}} z_j$. On the other hand:
\[
z_k' \leq \max_{j\in\{1,\ldots,k-1\}} z'_j \leq \max_{j\in\{1,\ldots,k-1\}} z_j = z_k.
\]

\noindent {\it Case 3} $ \max_{j\in\{1,\ldots,k-1\}} z_j \geq \theta_k > \max_{j\in\{1,\ldots,k-1\}} z'_j$. Then $z_k = \max_{j\in\{1,\ldots,k-1\}} z_j$ and $z'_k = \theta_k$.
\HB
\II

The above claim does not hold if we also take the last player into
account.  The reason is that when the last player's type equals the
currently highest bid, then he is indifferent between winning the item
or not and the set of his optimal strategies is the whole type space
$\Theta$.

To complete the proof now, note that the above claim implies that the
second highest bid among the first $n-1$ bids is at least as big in
$s(\cdot)$ as in $s'(\cdot)$.  But since we are in the case that
$\theta_n > \max_{j\in\{1,\ldots,n-1\}} [s(\cdot),\theta]_j$, this
means that ${[s(\cdot),\theta]}^{\ast}_3 \geq
{[s'(\cdot),\theta]}^{\ast}_3$.
\end{proof}

This maximal final social welfare under $s(\cdot)$ equals 
\[SW(\theta, s(\cdot)) = \theta_i - \frac{2}{n} ({[s(\cdot),\theta]}^*_2 - {[s(\cdot),\theta]}^*_3)\]
where $i = \textrm{argsmax} \: \theta$. This is always greater than or
equal to the final social welfare achieved in a BC auction when
players bid truthfully, which is $\theta_i - 2/n(\theta^*_2 -
\theta^*_3)$.
Additionally, for some inputs, for example when the last
player is not the winner, it is strictly greater.  It is also strictly
greater when the last player is the winner but the last but one player
does not have the highest type among the first $n-1$ players.

\section{Comments on a Nash implementation}
\label{sec:nash}

In this paper we studied sequential mechanisms. Alternatively, we
could view them as simultaneous ones in which each player $i$ receives
a type $\theta_i \in \Theta_i$ and subsequently submits a function
$r_i: \Theta_1 \times \LL \times \Theta_{i-1} \myra \Theta_i$ instead
of a type $\theta'_i \in \Theta_i$.  The submissions are simultaneous.
So the behaviour of player $i$ can be described by a strategy $s_i:
\Theta_1 \times \LL \times \Theta_i \myra \Theta_i$ which when applied
to the received type $\theta_i$ yields the function $s_i(\cdot,
\theta_i): \Theta_1 \times \LL \times \Theta_{i-1} \myra \Theta_i$
that player $i$ submits.  Define then 
\[
\mbox{$s(\cdot) \succeq_i s'(\cdot)$ iff for all $\theta \in \Theta$,
   $v_i(f([s(\cdot), \theta]), \theta_i) \geq v_i(f([s'(\cdot),
   \theta]), \theta_i)$).}
\]

We now say that a joint strategy $s(\cdot)$ is a \oldbfe{Nash
  equilibrium} if for all $i \in \{1, \LL, n\}$ and all 
strategies $s'_i(\cdot)$ of player $i$ we have
\[
(s_i(\cdot), s_{-i}(\cdot)) \succeq_i (s'_i(\cdot), s_{-i}(\cdot)).
\]

By Groves theorem for all sequential Groves
auctions the vector of truth-telling strategies $\pi_i(\cdot)$ is a
Nash equilibrium.  In contrast, the vector of strategies $s_i(\cdot)$
defined in (\ref{equ:opt}) is \emph{not} a Nash equilibrium in a
sequential Vickrey auction. Indeed, take two players and $\theta =
(1,2)$. Then for player 1 it is advantageous to deviate from
$s_1(\cdot)$ strategy and submit $3$. This way player 2 submits 0 and
player's 1 final utility becomes 1 instead of 0. The same remark holds
for sequential BC auctions.

On the other hand, if we only admit optimal strategies, then
Theorems \ref{thm:dom-opt} and \ref{thm:vic-opt} state that in the case
of sequential Vickrey auctions the vector of
strategies $s_i(\cdot)$ defined in (\ref{equ:opt}) is a dominant
strategy Nash equilibrium that is also Pareto optimal.
An analogous result holds for sequential BC auctions, though the qualification
`dominant' needs to be dropped. We omit the proof.

\section{Final remarks}
\label{sec:conc}

This paper
and our previous recent work, \cite{AE08}, 
forms part of a larger research endevour in which we seek in
sequential versions of commonly used incentive compatible mechanisms
optimal strategies that, when followed by all players, yield a maximal
social welfare. We studied sequential versions of two mechanisms in
the single unit auction setting: Vickrey auction and Bailey-Cavallo
mechanism.  We showed that in each of them natural optimal strategies
exist with the property that when each player follows them, a maximal
social welfare results. 

One could carry out a similar analysis for other Groves auctions with
relatively simple redistribution functions.  However, we are not aware
of a uniform approach to identify in Groves auctions optimal
strategies that maximize social welfare.  We believe that these
results can be extended to multi-unit auctions with unit demand
bidders.  A natural question is whether similar results can be
established for other types of auctions, for example auctions with
single-minded bidders or more general combinatorial auctions.

The maximal social welfare is attained here in the sense that any
vector of different optimal strategies yields a smaller social
welfare.  This is in contrast to the customary, simultaneous setting,
in which weaker results were established in \cite{ACGM08}, namely that
specific incentive compatible mechanisms are undominated.  Various
related questions remain open. For example, in the case of the public
project problem with unequal participation costs it is not known
whether undominated feasible (simultaneous) Groves mechanisms exist.

We would like to undertake similar study of the incentive compatible
mechanism proposed in \cite{NR01}, and of its sequential version. This
mechanism deals with the problem of purchasing a shortest path in a
network and is, in contrast to the mechanisms here considered, not
feasible.


\end{document}

%% file: ecmds.tex
\newcommand{\oldbfe}[1]{\begin{bfseries}\emph{#1}\end{bfseries}}

\newcommand{\myra}{\mbox{$\:\rightarrow\:$}}

\newcommand{\LL}{\mbox{$\ldots$}}

\newcommand{\C}[1]{\mbox{$\{{#1}\}$}}           

\newcommand{\NI}{\noindent}
\newcommand{\HB}{\hfill{$\Box$}}

\newcommand{\II}{\vspace{2 mm}}

\newcommand{\szkew}[1]{\relax \setbox0=\hbox{\kern -24pt $\displaystyle#1$\kern 0pt }%
\box0}
{\catcode`\@=11 \global\let\ifjusthvtest@=\iffalse}

\newcounter{oldmycaption}



%% file: header.tex
\def\smallromani{\renewcommand{\theenumi}{\roman{enumi}}
\renewcommand{\labelenumi}{(\theenumi)}}

\newtheorem{theorem}{Theorem}[section]
\newtheorem{defined}[theorem]{Definition}

\newtheorem{exa}[theorem]{Example}

\newtheorem{lemma}[theorem]{Lemma}
\newtheorem{corollary}[theorem]{Corollary}

\newtheorem{claim}{Claim}
\newtheorem{exe}{Exercise}

\newtheorem{pro}{Problem}

\newcounter{symbol}
\newcommand{\indexsyma}[1]%
{\stepcounter{symbol}\index{zzz1 \thesymbol @\protect#1}}
\newcommand{\indexsymb}[1]%
{\stepcounter{symbol}\index{zzz2 \thesymbol @\protect#1}}
\newcommand{\indexsymc}[1]%
{\stepcounter{symbol}\index{zzz3 \thesymbol @\protect#1}}
\newcommand{\indexsymd}[1]%
{\stepcounter{symbol}\index{zzz4 \thesymbol @\protect#1}}
\newcommand{\indexsyme}[1]%
{\stepcounter{symbol}\index{zzz5 \thesymbol @\protect#1}}